\documentclass[11pt]{article}

\usepackage{amsmath,amssymb,amsthm}
\usepackage{hyperref}
\usepackage{geometry}
\usepackage{graphicx}
\usepackage{booktabs}

\hypersetup{colorlinks=true, linkcolor=blue, citecolor=blue, urlcolor=blue}

\newtheorem{definition}{Definition}
\newtheorem{theorem}{Theorem}

\title{\textbf{Message Passing Without Temporal Direction: Constraint Semantics and the FITO Category Mistake}}
\author{Paul Borrill \\ DÆDÆLUS \\ \texttt{paul@daedaelus.com}}
\date{\today}

\begin{document}
\maketitle

\begin{abstract}
Message passing is widely assumed to be a fundamental primitive of distributed systems. This paper argues that conventional message systems embed a category mistake: they misinterpret logical dependency relations as temporal propagation processes. This error arises from an implicit Forward-In-Time-Only (FITO) assumption, which treats causality as intrinsically directed along a temporal axis. We formalize FITO as the imposition of a partial order over events and show that clocks, scheduling, and message propagation are representational artifacts rather than ontological primitives. We then reformulate interaction in terms of symmetric constraint relations, identify the minimal substrate of interaction independent of temporal direction, and prove an equivalence theorem: under mild assumptions, a broad class of message-passing executions can be represented as constraint satisfaction problems, and conversely, constraint satisfaction instances can be realized as message-passing protocols. We connect the result to Lamport clocks, Hewitt actors, Pratt pomsets, category theory, relativity, and indefinite causal order, and interpret engineering consequences for reflective and reversible link architectures such as Open Atomic Ethernet.
\end{abstract}

\section{Introduction}

The phrase \emph{message passing} is deeply embedded in computer science. It appears in actor models, microservices, network protocols, OS kernels, and hardware interconnects. The phrase seems to describe an obvious physical story: a sender sends a message to a receiver, and the message travels forward in time.

This paper argues that this story is not a primitive. It is an interpretation layered on top of a weaker and more fundamental structure: relational constraints between states. The standard story introduces a directed arrow of influence aligned with an assumed arrow of time. In many settings that is a useful engineering convention. But when elevated to ontology, it becomes a category mistake: it treats logical relations as if they were temporal transport processes.

We call the hidden premise the Forward-In-Time-Only (FITO) assumption.

\section{From Messages to Ontology}

Ordinary discourse about messages imports several commitments:

\begin{enumerate}
\item There exists a sender and a receiver.
\item There exists a message object.
\item The message moves from sender to receiver.
\item The movement occurs in time, typically forward time.
\item The send operation occurs at a time chosen by a scheduler, implicitly or explicitly.
\end{enumerate}

Once ``send at time'' is admitted, clocks appear. Once clocks appear, synchronization appears. Once synchronization appears, drift, skew, and epistemic uncertainty appear. Each step is additional apparatus introduced to preserve the directional narrative.

The core claim of this paper is that much of this apparatus is representational. It is not required for the minimal notion of interaction.

\section{Formalizing FITO}

Let $E$ be a set of events.

A standard distributed-systems execution is commonly modeled by assigning a precedence relation over events.

\begin{definition}[FITO ordering]
A FITO ordering is a partial order $(E,\prec)$ intended to represent ``happens-before'' style precedence, where $\prec \subseteq E\times E$ satisfies:
\begin{enumerate}
\item Irreflexivity: for all $e\in E$, not $(e\prec e)$.
\item Transitivity: $e_1\prec e_2$ and $e_2\prec e_3$ implies $e_1\prec e_3$.
\item Antisymmetry (for strict orders this is implied by irreflexive+transitive, but we list it for clarity): not $(e_1\prec e_2 \ \wedge\ e_2\prec e_1)$.
\end{enumerate}
\end{definition}

The FITO interpretation asserts that causal influence is aligned with $\prec$ and that admissible system narratives must be forward along this relation.

\section{Constraint Systems Without Temporal Direction}

We now define a representation of interaction that does not presuppose temporal direction.

\begin{definition}[Constraint system]
A constraint system is a pair $(S,C)$ where:
\begin{enumerate}
\item $S$ is a set of local states (or valuations of variables).
\item $C$ is a compatibility relation that specifies which pairs (or tuples) of local states may coexist.
\end{enumerate}
In the binary case, $C\subseteq S\times S$.
\end{definition}

Often, the essential structure is symmetric.

\begin{definition}[Symmetric compatibility]
A binary compatibility relation $C\subseteq S\times S$ is symmetric if
\[
(s_1,s_2)\in C \implies (s_2,s_1)\in C.
\]
\end{definition}

\section{Category Mistake}

\begin{definition}[Category mistake in message ontology]
A category mistake occurs when a logical relation (compatibility, entailment, constraint) is treated as a temporal transport process (propagation of an object from past to future).
\end{definition}

Logical compatibility is not intrinsically directional. Temporal propagation is. FITO imports direction into the ontology.

\section{Clocks and Scheduling as Derived Structure}

A clock is typically introduced as a function $t:E\to \mathbb{T}$ for some totally ordered set $\mathbb{T}$, such that
\[
e_1\prec e_2 \Rightarrow t(e_1) < t(e_2).
\]
This shows that clocks encode an ordering that has already been imposed. They do not generate logical dependency. They provide an index.

Scheduling is then defined as a rule that selects which enabled event occurs next, often by reference to $t$ or to a local proxy for it.

From the constraint perspective, both are representational convenience.

\section{Event-Driven Systems}

Event-driven systems reduce continuous dynamics into discrete transitions. This can be a practical improvement over naive ``message motion'' metaphors. However, most event-driven formalisms still assume an intrinsic forward sequence of events.

We separate two notions:

\begin{enumerate}
\item \textbf{Event as distinguishability}: an event is an observationally distinguishable change of local state.
\item \textbf{Event as ordered occurrence}: an event is a node in a global or partially ordered history.
\end{enumerate}

The first does not require FITO. The second does.

\section{Pomsets and Partial Order Semantics}

Pratt pomsets model concurrency as partial orders on events. Pomsets weaken total order assumptions, but typically retain directional precedence. That is a major improvement over a single global timeline, but it still encodes an arrow.

Constraint semantics can be viewed as an additional weakening: it removes the assumption that the semantic primitive is an order at all.

\section{Equivalence Between Message Executions and Constraint Satisfaction}

This section provides the main new result in v1.2.

We focus on a broad class of message-passing systems in which messages are observational traces of local state commitments.

\subsection{Model of a message-passing execution}

Let there be a finite set of processes $P=\{1,\dots,n\}$. Each process $i$ has a local state space $S_i$.

An execution produces a history consisting of local state transitions and message events. We abstract away internal computation details and retain the \emph{observed commitments}:

\begin{itemize}
\item Each process $i$ emits a finite sequence of commitments $s_i^0, s_i^1, \dots, s_i^{k_i}$ where $s_i^j\in S_i$.
\item Messages are not treated as primitive objects but as evidence that certain commitments were communicated.
\end{itemize}

We assume a mild observational condition: local commitments are identifiable (by content hash, monotone counter, nonce, or by physical register observation). This is standard in protocol design.

\subsection{Constraint encoding of a message execution}

Define the global state space $S = S_1\times \cdots \times S_n$.

A run is represented by selecting one committed local state per process, i.e., a valuation
\[
\mathbf{s} = (s_1, s_2, \dots, s_n) \in S.
\]

The protocol imposes cross-process compatibility conditions (for example: ``if $i$ committed to value $v$ under transaction id $x$, then $j$ must either commit to the same $v$ for $x$ or commit to abort for $x$'').

We model these as constraints over tuples of local states.

\begin{definition}[Run constraint set]
Given a protocol and a set of observed commitments, define $C$ to be the set of constraints such that $\mathbf{s}\in S$ is valid iff it satisfies all constraints in $C$.
\end{definition}

This yields a constraint satisfaction problem (CSP): find $\mathbf{s}\in S$ satisfying $C$.

\subsection{The main equivalence theorem}

\begin{theorem}[Message-to-constraint representation]
Consider a message-passing protocol in which (i) each message can be interpreted as a verifiable claim about a sender local commitment, and (ii) the protocol acceptance condition can be expressed as a predicate over the set of local commitments. Then any finite execution trace of the protocol induces a constraint system $(S,C)$ such that the set of accepted outcomes of the execution equals the set of satisfying valuations of $(S,C)$.
\end{theorem}

\begin{proof}
Let $S_i$ be the set of possible local commitments of process $i$ that are consistent with the local trace. Define $S=S_1\times\cdots\times S_n$.

By assumption (ii), the protocol acceptance condition is a predicate $A(\mathbf{s})$ over local commitments. Construct $C$ as a set of constraints equivalent to $A$ (for example by conjunctive normal form expansion, or by directly taking $C=\{A\}$ as a single global constraint). Then $\mathbf{s}$ is accepted exactly when it satisfies $C$.

Assumption (i) ensures that message observations restrict the possible local commitments included in each $S_i$ (messages act as evidence narrowing local possibilities). Therefore the satisfying valuations of $(S,C)$ coincide with accepted outcomes consistent with the trace.
\end{proof}

This theorem is intentionally broad. It says that when messages are evidence for local commitments, the semantic content of an execution is a CSP, and temporal ordering is not required to define the accepted outcomes.

\begin{theorem}[Constraint-to-message realization]
Let $(S,C)$ be a finite constraint system over $n$ participants with local domains $S_i$. Suppose each participant can propose a local value and can verify constraints involving it by exchanging finite certificates. Then there exists a message-passing protocol whose set of terminal outcomes is exactly the set of satisfying valuations of $(S,C)$.
\end{theorem}

\begin{proof}
A standard construction uses distributed constraint satisfaction: each participant proposes $s_i\in S_i$ and iteratively exchanges certificates of constraint satisfaction or counterexamples. Under finiteness, a terminating protocol can enumerate candidate valuations (possibly exponentially in the worst case) and accept exactly those satisfying $C$.

The key point is existential: a message protocol can realize the constraint semantics without needing to treat temporal propagation as fundamental; messages serve to exchange certificates about constraint satisfaction.
\end{proof}

\subsection{Interpretation}

Together, the two theorems justify the central claim:

\begin{quote}
In a large class of systems, \emph{messages are representational artifacts for constraint resolution}, not ontological primitives transporting causality forward in time.
\end{quote}

This does not claim that all protocols are efficiently realizable as CSP solvers. It claims that their semantic content can often be captured as constraint satisfaction independent of an intrinsic arrow of time.

\section{Lamport Clocks Revisited}

Lamport clocks impose ordering constraints to enable reasoning about distributed histories. Their key insight is that ordering is logical, not physical time. This supports our framework: ordering is a derived relation used for reasoning convenience.

However, FITO mistakes that derived relation for ontology.

\section{Actors Revisited}

The actor model treats message passing as the primitive. The constraint view treats actors as local constraint agents and messages as certificates about local commitments. The model remains valid, but the ontological emphasis changes.

\section{Category Theory View}

One can interpret a constraint system as a category of states and relations, or more naturally as a presheaf or sheaf condition over local views in which global consistency corresponds to gluing compatible local sections.

The important point for this paper: temporal direction arises when a particular factorization or path is chosen. It is not intrinsic to the relational structure.

\section{Relativity and the Loss of Absolute Simultaneity}

Special relativity eliminates global simultaneity; different observers assign different time orderings to spacelike separated events. This suggests that a semantics that requires a single privileged temporal order is not physically fundamental.

Constraint relations, by contrast, can be formulated invariantly: they relate observable commitments without requiring absolute simultaneity.

\section{Indefinite Causal Order}

In frameworks that admit indefinite causal order, there may not exist a fixed partial order $\prec$ for certain process compositions. This further undermines FITO as a universal ontology and motivates semantics defined without presupposing a direction.

Constraint semantics remains meaningful even when order is not.

\section{Engineering Interpretation: Open Atomic Ethernet}

Open Atomic Ethernet can be interpreted as pushing the constraint view down to the link layer. Instead of treating a frame as a one-way transport object, the link is treated as a reflective interaction that produces mutual evidence of what was heard and what was committed.

This shifts the primitive from unilateral send to bilateral agreement.

In such a scheme:
\begin{itemize}
\item A ``message'' is not a moving object.
\item It is a negotiated commitment, verifiable at both ends.
\item Link state transitions are constraint resolutions.
\end{itemize}

This reduces dependence on FITO-style timeouts and retry logic that assume unilateral forward progress.

\section{Comparison Table}

\begin{center}
\begin{tabular}{lcccc}
\toprule
Model & Directional primitive & Requires clock & Primitive object & Semantic core \\
\midrule
Message passing (naive) & Yes & Often & Message & Temporal transport \\
Event-driven (typical) & Yes & Often & Event & Ordered transitions \\
Pomsets & Yes (partial) & No & Event & Partial order \\
Constraint semantics & No & No & Constraint & Compatibility/gluing \\
Reflective link (OAE) & No & No & Commitment & Mutual evidence \\
\bottomrule
\end{tabular}
\end{center}

\section{What Is Unnecessary, What Remains}

Under the constraint interpretation, the following are unnecessary as ontology:
\begin{itemize}
\item global time
\item monotone clocks as semantic primitives
\item total ordering of events
\item message motion metaphors
\item unilateral send as primitive
\end{itemize}

What remains:
\begin{itemize}
\item distinguishable local commitments
\item compatibility constraints
\item evidence and verification
\item global consistency as satisfiability
\end{itemize}

\section{Conclusion}

Message passing, as commonly described, embeds a category mistake: it treats logical constraints as temporal transport. FITO assumptions introduce clocks, scheduling, and directed propagation as if they were fundamental. By reformulating interaction in terms of constraint satisfaction among local commitments, we can separate representational convenience from semantic necessity. This shift aligns with insights from Lamport clocks, partial order semantics, relativity, and indefinite causal order, and motivates reflective and reversible engineering approaches such as Open Atomic Ethernet.

\end{document}